\documentclass[sn-basic]{sn-jnl}% Basic Springer Nature Reference Style/Chemistry Reference Style
\usepackage{eso-pic}
\usepackage{hyperref}
\usepackage{graphicx}
\usepackage{amsmath}
\usepackage{amsfonts}
\usepackage{array}
\usepackage{dblfloatfix}
\usepackage{grffile}
\usepackage{float}
\usepackage{caption}
\usepackage{subcaption}
\usepackage{upgreek}
\usepackage{multirow}
\usepackage{physics}
\usepackage{color}
\usepackage{algorithm}
\usepackage{algpseudocode}

\usepackage[nameinlink,capitalize]{cleveref}

\usepackage[normalem]{ulem}               % to striketrhourhg text
\newcommand\redout{\bgroup\markoverwith
{\textcolor{red}{\rule[0.5ex]{2pt}{0.8pt}}}\ULon}

\definecolor{dgreen}{RGB}{0, 160, 0}
\newcommand{\revg}[1]{\textcolor{dgreen}{#1}}

\jyear{2023}%

\theoremstyle{thmstyleone}%
\newtheorem{theorem}{Theorem}%  meant for continuous numbers
\theoremstyle{thmstyletwo}%
\newtheorem{remark}{Remark}%

\theoremstyle{thmstylethree}%

\begin{document}

\AddToShipoutPictureFG{%
\put(0,20){
\hspace*{\dimexpr0.075\paperwidth\relax}%% change \dimexpr0.5\paperwidth\relax appropriately
\parbox{.84\paperwidth}{\footnotesize The Version of Record of this article is published in  Journal of Control, Automation and Electrical Systems, and is available online at   \url{https://doi.org/10.1007/s40313-023-01041-1}}%
}}

\title[Autonomous driving of trucks in off-road environment]{Autonomous driving of trucks in off-road environment}

\author*[1]{\fnm{Kenny A. Q.} \sur{Caldas}}\email{kennycaldas@usp.br}

\author[2]{\fnm{Filipe M.} \sur{Barbosa}}\email{filipe.barbosa@liu.se}

\author[3]{\fnm{Junior A. R.} \sur{Silva}}\email{junior.anderson@usp.br}

\author[3]{\fnm{Tiago C.} \sur{Santos}}\email{tiagocs@icmc.usp.br}

\author[3]{\fnm{Iago P.} \sur{Gomes}}\email{iagogomes@usp.br}

\author[3]{\fnm{Luis A.} \sur{Rosero}}\email{lrosero@usp.br}

\author[3]{\fnm{Denis F.} \sur{Wolf}}\email{denis@icmc.usp.br}

\author[1]{\fnm{Valdir} \sur{Grassi Jr}}\email{vgrassi@usp.br}

\affil*[1]{\orgdiv{Department of Electrical and Computer Engineering}, \orgname{S\~ao Carlos School of Engineering, University of S\~ao Paulo}, \city{S\~ao Carlos}, \country{Brazil}}

\affil[2]{\orgdiv{Division of Automatic Control, Department of Electrical Engineering}, \orgname{Linköping University}, \city{Linköping}, \country{Sweden}}

\affil[3]{\orgdiv{Institute of Mathematics and Computer Science}, \orgname{University of S\~ao Paulo}, \city{S\~ao Carlos}, \country{Brazil}}


\abstract{
Off-road driving operations can be a challenging environment for human conductors as they are subject to accidents, repetitive and tedious tasks, strong vibrations, which may affect their health in the long term. Therefore, they can benefit from a successful implementation of autonomous vehicle technology, improving safety, reducing labor costs and fuel consumption, and increasing operational efficiency. 
The main contribution of this paper is the experimental validation of a path tracking control strategy, composed of longitudinal and lateral controllers, on an off-road scenario with a fully-loaded heavy-duty truck.
The longitudinal control strategy relies on a Non-Linear Model Predictive Controller (NMPC), which considers the path geometry and simplified vehicle dynamics to compute a smooth and comfortable input velocity, without violating the imposed constraints. The lateral controller is based on a Robust Linear Quadratic Regulator (RLQR), which considers a vehicle model subject to parametric uncertainties to minimize its lateral displacement and heading error, as well as ensure stability. Experiments were carried out using a fully-loaded vehicle on unpaved roads in an open-pit mine. The truck followed the reference path within the imposed constraints, showing robustness and driving smoothness.
}
\keywords{Autonomous truck, Path tracking control, Robust control, Autonomous driving.}

\maketitle

%%===================================================%%
%%                    Introduction      %%
%%===================================================%%
\section{Introduction}
\label{sec:INTRODUCTION}

Autonomous truck in urban areas is still a major challenge in the field of intelligent transportation since it has to deal with many external factors, such as pedestrians, cyclists, other vehicles, traffic signals, and so on \citep{Ljungqvist2020}. On the other hand, off-road applications are more suitable to the initial deployment of this technology. In this context, mining is one of the sectors that can benefit from the implementation of autonomous trucks.

According to \cite{Kecojevic2004} and \cite{Zhang2014}, there are several truck-related accidents and fatalities in mining areas, due to loss of control, rollovers, reverse maneuvering and collision, to name a few. Moreover, drivers are subjected to repetitive and tedious tasks as well as strong vibrations, which can affect their health in the long term \citep{Lima2017,Xiao2020}. Therefore, the adoption of autonomous trucking in this context could improve safety, vehicle control, and reduce labor costs. From an economic standpoint, it is also possible to reduce fuel consumption and increase operational efficiency, since the daily operating time would be longer \citep{Slowik2018}.

This paper presents a lateral and longitudinal controllers of an autonomous truck for a GNSS-based driving system. The control strategy is the combination of our previous works,  a Nonlinear Model Predictive Controller (NMPC) \citep{Caldas2019} and a Robust Linear Quadratic Regulator (RLQR) \citep{Barbosa2019}. The former deals with the vehicle's longitudinal control, optimizing its speed subject to path, safety and stability constraints, while the latter deals with the vehicle's lateral control subject to mass uncertainty. Both aforementioned works  \citep{Caldas2019,Barbosa2019} were only evaluated on simulations. The main contribution of this paper is the experimental validation of the control strategy on an off-road scenario with a fully-loaded truck in order to evaluate its performance and limitations. The paper also briefly presents the system architecture in which the control strategy was implemented and tested.

This paper has the following structure: \autoref{sec:RELATED_WORK} presents other studies related to autonomous truck control strategies with their main contributions and results. \autoref{sec:architecture} presents an overview of the system architecture with an emphasis on the vehicle control main modules and a short description of their operations. The longitudinal and lateral control strategies are described in \autoref{sec:long_ctrl} and \ref{sec:lat_ctrl}, respectively. The experiments and obtained results are shown in \autoref{sec:exp_res} and subsequently discussed in \autoref{sec:discussion}. Lastly, conclusions are drawn in \autoref{sec:conclusion}.

%%===================================================%%
%%                    Related Workd      %%
%%===================================================%%
\section{Related work}
\label{sec:RELATED_WORK}

Autonomous navigation of trucks is a complex task in the field of intelligent transportation systems. In off-road scenarios, the dynamics of the heavy-duty truck are more susceptible to disturbances and modeling errors than light vehicles, due to its complexity. Consequently, the adopted control strategies must be robust enough to overcome them in the driving task. For that reason, experimental evaluations in off-road autonomous truck navigation systems are still very challenging.

Currently, the majority of state-of-the-art works focus on lateral and longitudinal controllers separately. A fuzzy-inference self-tuning steering control for heavy-duty vehicles under freight weight change was proposed by \cite{Ario2016} to avoid a decrease in the controller's performance when the system parameters are varying. In this case, the controller was evaluated only in simulation. Similarly, a Takagi-Sugeno fuzzy controller is implemented in a 13 tons truck to control the vehicle in different driving conditions by \cite{Rodriguez-Castano2016}. In the experiment, the vehicle followed a path with accentuated curves with a maximum track-error of 71 cm, where its velocity was controlled by a human supervisor.

Recent researches have used Model Predictive Control (MPC) for navigation. For instance, an algorithm for lateral control of an autonomous truck for smooth and accurate steering was presented by \cite{Lima2017}, where the main difference from a standard MPC approach is that the driving smoothness is considered in the cost function. An experiment to evaluate the controller's performance was performed at Scania's facilities in two different paths: a precision track and a high-speed test track, in order to resemble mining and highway scenarios, respectively. A later version of this work is presented by \cite{Lima2018}, where input rate constraints and a new approach for computing the terminal cost of the MPC controller are proposed and evaluated experimentally in a scenario that resembles an emergency maneuver. MPC has also been used for longitudinal control of heavy trucks, as proposed by \cite{Held2019}. In this work, the authors could reduce almost 7 \% of energy consumption in comparison with a normal driving pattern, in a simulated environment.

Differently from the aforementioned works, the solution presented in this paper considers both longitudinal and lateral controllers for the path tracking problem. Also, the presented system is evaluated in a real mining scenario with a fully-loaded vehicle, which to the best of our knowledge has not been previously addressed.

%%===================================================%%
%%                    Architecture     %%
%%===================================================%%
\section{System architecture}
\label{sec:architecture}

The heavy-duty truck used in the experiments is the Scania P-420 CB8x4ESZ, shown in \autoref{fig:scania}. This model has 12 gears, operating with a manual transmission at first gear and with an automatic transmission with the others, therefore, a human driver must accelerate and change to second gear before starting the experiments with the autonomous mode.
For safety reasons the truck was equipped with a panel with some switches to turn  the autonomous system on/off.

%--
\begin{figure}[htb]
	\centering
	\includegraphics[width=0.9\linewidth]{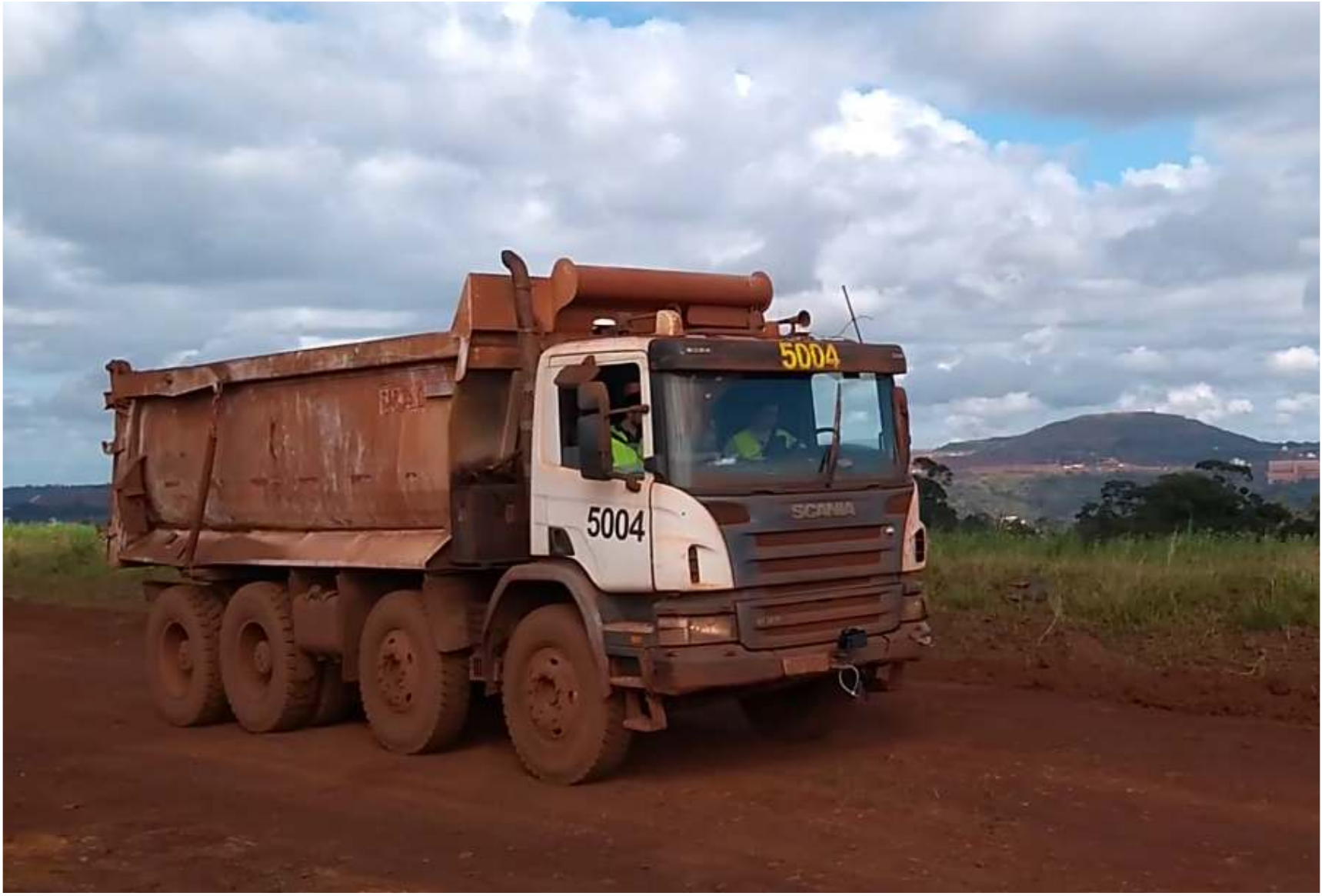}
	\caption{Scania P-420 CB8x4ESZ truck.}
	\label{fig:scania}
\end{figure}
%--

The system architecture used to autonomously control the vehicle is presented in \autoref{fig:architecture}. Furthermore, the developed software is based on the Robot Operating System (ROS) Kinetic, running on Ubuntu 16.04. Since the main focus of this paper is the high-level controllers, only a brief description of the other blocks will be given in this section. 

%--
\begin{figure*}[htb]
    \centering
    \includegraphics[width=0.9\linewidth]{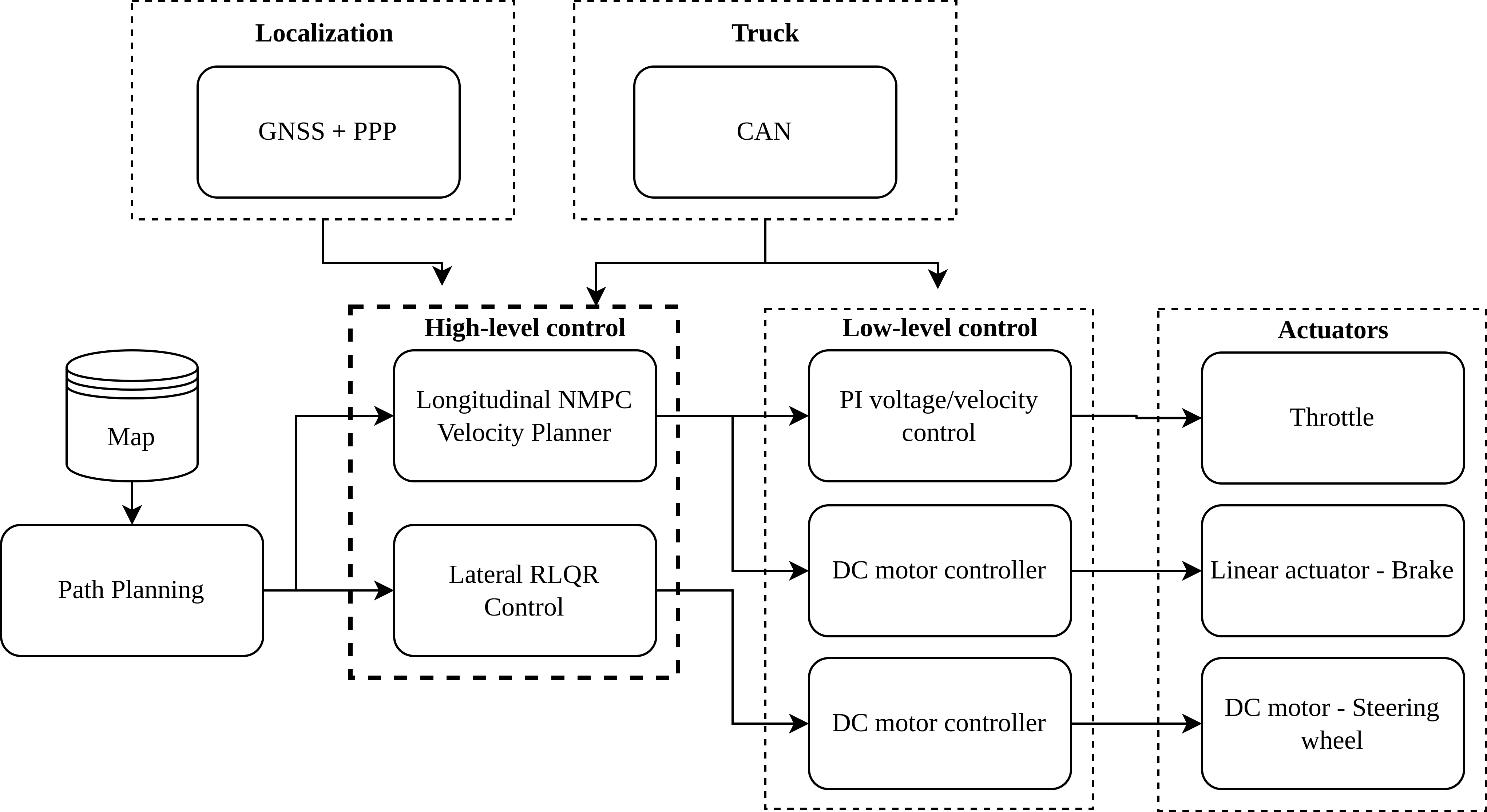}
    \caption{Autonomous truck system architecture.}
    \label{fig:architecture}
\end{figure*}
%--

\subsection{Localization}
\label{subsec:loc}

The autonomous system uses a Global Navigation Satellite System (GNSS) based localization system with Precise Point Positioning (PPP) correction that provides centimeter-level positioning accuracy. Although other experiments and research efforts have noticed that GNSS navigation alone can be very noisy, the precision of position estimation with the dedicated correction system and the conditions of the site where the experiments were conducted allowed navigation using waypoints (i.e., position and heading). Moreover, as an additional safety measure, a diagnostic system is used to monitor the signal quality of the GNSS receiver, which monitors the correction signal, number of satellites, and other parameters that ensure the quality of the information provided by the receiver.

The localization module runs at 10 Hz frequency and was responsible for obtaining geodetic  coordinates and converting it into Universal Transverse Mercator (UTM) projection, which is a Cartesian coordinate system. The receiver also provides data on heading direction, based on a two-antenna system.

\subsection{Map}
\label{subsec:mapping}

The map is composed of routes segments represented by sequences of waypoints recorded offline during previous manual driving. The information recorded for each waypoint includes the position, heading, and altitude, obtained by the GNSS-based localization system with PPP correction (see Section \ref{subsec:loc}). Each route segment is stored in a file for further access by the path planning.

\subsection{Truck}

The current state of the vehicle is estimated by a module that uses data from the CAN (Controller Area Network) bus and sensors coupled to the actuators. The state message carries information on velocity, current gear and engine speed, steering angle, and others. Moreover, the update frequency is 5 Hz, and this information feed all high-level and low-level controllers. 

\subsection{Actuators}\label{subsec:actuators}

There are three actuators in the proposed system: the brake, steering wheel, and the truck's throttle. A linear actuator adjusts the brake pedal position and a DC motor attached to the steering column adjusts the steering wheel angle, each one controlled by a separate DC motor controller. The throttle pedal position, ranging from $0-100\%$, sends a corresponding voltage signal to the truck’s electronic injection to accelerate the truck using measurements obtained from two potentiometers placed on the pedal. Thus, for autonomously controlling truck's throttle, this voltage signal is emulated. For this purpose, a microcontroller is used to control the truck’s velocity using its analog voltage output, accordingly to the vehicle's specifications.

\subsection{Low-level controller}

The low-level controllers are responsible for converting the reference setpoints, determined by the high-level controllers, into actuators' input signals. The longitudinal NMPC controller generates a desired reference velocity profile (see \autoref{sec:long_ctrl}). Whenever the current vehicle's velocity is below the desired reference, i.e. the vehicle needs to accelerate, a  PI low-level controller converts the reference velocity into a voltage reference for the microcontroller discussed in the previous subsection. Whenever the vehicle needs to decelerate, the desired braking signal is set as reference for the DC motor controller responsible for acting on the linear actuator attached to the brake pedal. Another DC motor controller is utilized to set the steering wheel to the desired steering angle defined as reference by the lateral RLQR controller (see \autoref{sec:lat_ctrl}).

\subsection{Path planning}

In order to execute a navigation task, the path planning system creates a reference path by selecting and concatenating the waypoints of corresponding routes segments (i.e., files).
These waypoints were recorded using a GNSS-based localization system with PPP correction, which provides centimeter-level of accuracy. However, this
data can still be noisy enough to cause abrupt changes in the position, orientation, or curvature of the road segment, which can lead to dangerous situations or uncomfortable driving experiences (e.g., steering wheel oscillation). Thus, the original sequence of waypoints must be properly processed and smoothed.

Smooth paths that resemble the original reference can be computed by interpolating curves such as cubic splines, B-splines, B\'ezier and clothoid curves. These are commonly applied to road shape modeling, since they are smooth and can be parameterized by control points \citep{Gonzalez_2016}. In this work, we use clothoids to represent the reference path. Clothoids are curves whose curvature varies linearly with respect to their arc lengths, providing a constant curvature rate change \citep{Silva_2018}, which makes them suitable for railways and highways design \citep{Aashto_2001}. In addition, this curve can be parameterized by its curvature, taking into account nonholonomic constraints and comfort level.
For this reason, the path planning module processed the noisy waypoints previously recorded to create  piecewise linear continuous-curvature paths composed of clothoids according to the methodology developed by \cite{Silva2020}.

%%===================================================%%
%%                    Long Control     %%
%%===================================================%%
\section{Longitudinal NMPC velocity planner}
\label{sec:long_ctrl}

The main advantages of the NMPC strategy are the possibility of optimizing the vehicle's speed based on the physical characteristics of the path, such as altitude and curvature, while satisfying constraints on the state and control in the solution space. As illustrated in \autoref{fig:architecture}, the information collected from the CAN bus provides the vehicle speed in meters per second (m/s), engine speed in Rotations Per Minute (RPM) and current gear. The global position of the vehicle is updated with the GNSS sensor, to find the truck's position with respect to the path map.

The solution of the NMPC optimization problem is solved by using the C/GMRES algorithm, which is a combination of the continuation method with a linear equation solver called generalized minimum residual (GMRES) method \citep{ohtsuka2004}. Instead of using computationally expensive methods such as Riccati differential equations, the system’s optimal trajectory can be calculated using a Hamiltonian function $H$, composed of the state equation $f(x,u)$, cost functions $J(x,u)$ and equality constraints $C(x,u) \in \mathbb{R}^{m_c}$, defined by
\begin{align}
    H(x,\lambda,u,\mu) = J(x,u) + \lambda^T f(x,u) + \mu^T C(x,u),
\end{align}
where $x(t) \in \mathbb{R}^n$ is the state vector and $u(t) \in \mathbb{R}^{m_u}$ the control input vector. $\lambda$ is the costate and $\mu$ the Lagrange multiplier associated with the equality constraints.

In this formulation, the necessary first-order conditions for optimality are calculated as
\begin{align}
    &x^*_\tau = f(x^*,u^*), \\
    &x^*_0(t) = x(t), \\
    &\lambda^*_\tau = - H_x^T(x^*, \lambda^*, u^*, \mu^*), \\
    &H_u(x^*, \lambda^*, u^*, \mu^*) = 0, \label{eq:cond_hu} \\
    &C(x^*, u^*) = 0, \label{eq:cond_constraint}
\end{align}
where the subscript $\tau$ represents a partial differentiation with respect to a fictitious time $\tau$. Thus, it is possible to define a vector $U(t)$ composed of the control input vectors and the Lagrange multipliers over the prediction horizon as
\begin{align}
    U(T) = \left[ u_0^{*T}(t), \mu_0^{*T}(t), \ldots, u_{N-1}^{*T}(t), \mu_{N-1}^{*T}(t) \right]^T \in \mathbb{R}^{mN},
\end{align}
where size $m = m_u + m_c$. Hence, a $mN$-dimensional system can be defined to combine the optimality conditions in a single equation, given by
\begin{align}
    F(U(t), x(t)) = 
    \begin{bmatrix}
        H_u^T(x_0^{*}, \lambda_1^{*}, u_0^{*}, \mu_0^{*}) \\
        C(x_0^{*}, u_0^{*}) \\
        \vdots \\
        H_u^T(x_{N-1}^{*}, \lambda_N^{*}, u_{N-1}^{*}, \mu_{N-1}^{*}) \\
        C(x_{N-1}^{*}, u_{N-1}^{*}) \\
    \end{bmatrix} 
    = 0.
\end{align}
Therefore, the control input can be calculated by
\begin{align} \label{eq:ControlCalc}
     \dot{U} = F_{U}^{-1}(-F_{x}-F_{t}- \zeta F),
\end{align}
where $\zeta$ is a constant scalar. See the paper of \cite{ohtsuka2004} for more details on this method. 

\subsection{Definitions and vehicle longitudinal dynamics}

The longitudinal states variables of the truck are described as
\begin{align}\label{eq:StateVector}
   x(t) = [s(t),v(t)]^{T},
\end{align}
where $s(t)$ is the traveled distance and $v(t)$ the vehicle velocity. The traction force $F_T(t)$, given in Newton per kilogram (N/kg) is defined as the control input by
\begin{align}
  F_{T}(t) &= m_1u(t), \label{eq:TractionForce}
\end{align}
where positive and negative values correspond to throttle and brake pedals actuation, respectively.

The acting forces on the truck are illustrated in \autoref{fig:DynamicModel}. The resistance forces are given by $F_{grav}$, $F_{drag}$ and $F_{roll}$, which are the gravitation, the aerodynamic drag and rolling resistance force, respectively. The summation of all forces is represented as
\begin{align}\label{eq:Forces}
    m_1\dot{v}(t) = F_{T}(t)-F_{roll}(t)-F_{drag}(t)-F_{grav}(t),
\end{align}
with
\begin{subequations}\label{eq:ResForces}
    \begin{align}
        F_{roll}(t) &= C_{rr}(v) m_1g \cos(\beta(s)), \label{eq:RollForce} \\
        F_{grav}(t) &= m_1g \sin(\beta(s)),\label{eq:GravForce} \\
        F_{drag}(t) &= \frac{1}{2} C_{D} \sigma A_{f} v^{2},\label{eq:DragForce}
    \end{align}%
\end{subequations}
where $m_1$ is the truck mass, $C_{D}$ the aerodynamic drag coefficient, $\sigma$ the air density, $A_{f}$ the frontal area of the vehicle and $\beta(s)$ the road slope angle at position $s$. The rolling resistance coefficient is given by
\begin{align}\label{eq:RollCoef}
    Crr(v) = 0.01(1+\revg{\mid v\mid}/576).
\end{align}
\begin{figure}[ht]
    \centering
    \includegraphics[width=0.6\linewidth]{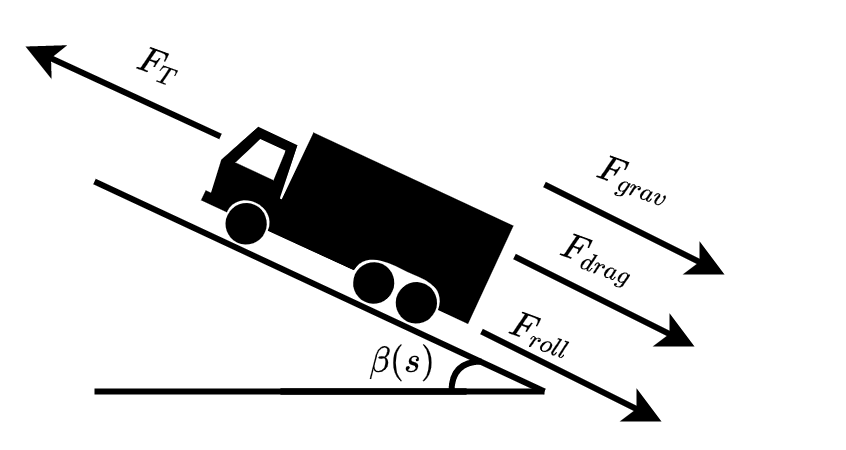}
    \caption{\label{fig:DynamicModel} Acting forces on the vehicle.}
\end{figure}

Thus, the state equation of the system can be defined as
\begin{align}
    f(x,u) = 
    \begin{bmatrix}
    \dot{s}\\
    \dot{v} \end{bmatrix} = 
        \begin{bmatrix}
            v \\
            u - \frac{1}{2m_1}C_{D} \rho A_{f} v^{2} - g \sin(\beta(s)) - C_{rr}(v) g \cos(\beta(s))
        \end{bmatrix},
\end{align}
where the vehicle acceleration is represented by $\dot{v} = a$ later in this work. We highlight that we only consider the forward motion of the truck for the state space representation and the NMPC control design.

The path slope and curvature information can be obtained from the path planning module. This is vital to ensure a safe and stable trip during the velocity planning step. The slope angle, given in radians, can be calculated as
\begin{align}\label{eq:RoadGrade}
 \beta(s) = \tan^{-1} \left( \frac{R_{alt}(s+\Delta_s)-R_{alt}(s-\Delta_s)}{2\Delta_s} \right).
\end{align} 
A user-defined distance $\Delta s$ is used to approximate the slope angle, considering the difference in altitude measured in two different coordinates. The truck’s current altitude is given by $R_{alt}$. Observe that when $\beta(s)$ is negative, i.e. the truck is moving downhill, the sine in \eqref{eq:GravForce} will be negative, hence, the gravitational resistance force will increase the traction force $F_T(t)$ instead of subtracting.

In curves, the truck is very susceptible to lateral acceleration depending on the path curvature and its current velocity. Before entering an accentuated curve section, the truck should slow down to maintain the vehicle’s stability. Thus, the lateral acceleration can be estimated by
\begin{align}\label{eq:LateralAcceleration}
    a_{lat} = v^{2}f_{curv}(s).
\end{align}
The path absolute curvature profile, at any position $s$, can be obtained as
\begin{align}\label{eq:CurvatureFunction}
    f_{curv}(s) = \abs{\frac{1}{R_{curv}(s)}},
\end{align}
where $R_{curv}$ is the radius of the curve.

The engine speed (RPM), is another important factor required to estimate the future states of the vehicle. Based on the current gear and vehicle speed, it is possible to calculate $\omega_{e}(t)$ as
\begin{align}\label{eq:EngSpeed}
	\omega_{e}(t) = \min \left( \max \left( \frac{1000v(t)\xi}{120 \pi r (1-i_{slip})},\omega_{idle} \right), \omega_{red} \right),
\end{align}
where $\xi$, $r$, $\omega_{idle}$, $\omega_{red} $ and $i_{slip}$ are the current gear ratio, wheel radius, engine speed at idle, redline engine speed and tire slippage percentage, respectively.

The vehicle's engine power in kW required to overcome the resistance forces, based on its current speed, is given by 
\begin{align}\label{eq:VehiclePower}
    P_{eng}(t) = \left( \frac{F_{res}(t) + m_1a(t)(1.04+0.0025\xi(t)^{2})}{3600\eta_{d}} \cdot v(t)\right).
\end{align}
The moment of inertia of the vehicle at different gear ratios $\xi$ is represented by the term $(1.04+0.0025\xi(t)^{2})$ and $\eta_{d}$ is the driveline efficiency \citep{wong2001}.

\subsection{Constraints}\label{subsec:constraints}

In order to avoid slippage during acceleration or braking, the input constraint given by 
\begin{align}\label{eq:IneqConstraint}
    -u_{max} \leq u \leq u_{max}, \quad u_{max} = \frac{\gamma M_{ta}g}{m_1},
\end{align}
is considered, where $M_{ta}$ is the mass of the propulsive axle in $kg$ and $\gamma$ the road adhesion. This inequality is based on the maximum frictional force that can be sustained between the vehicle's wheels on the propulsive axle, and the road surface \citep{SimpleModel}. Therefore, the slippage can be estimated based on the NMPC control input and the truck masses.

The inequality \eqref{eq:IneqConstraint} can be rewritten as an equality constraint using a heuristic treatment called Auxiliary Variable Method by introducing a slack input $u_{slk}$ as follows \citep{ohtsuka2004}:
\begin{align*}
    -u_{max} \leq &u \leq u_{max}  \\
    &\Downarrow \\
    u^2 &\leq u^2_{max} \\
    u^2 - u^2_{max} &\leq 0 \\
    u^2 + u^{2}_{slk} - u^2_{max} &= 0.
\end{align*}
By dividing everything by 2 to improve numerical stability during the Hamiltonian calculations, the equality constraint is defined as
\begin{align}\label{eq:EqConstraint}
    C(x(t),u(t)) =  \frac{(u^{2}+u_{slk}^{2}-u_{max}^{2})}{2}=0.
\end{align}

The slack input $u_{slk}$ is squared to prevent the sign of $u^{2}-u_{max}^{2}$ to be nonpositive. This step makes possible to rewrite the inequality constraint defined in \eqref{eq:IneqConstraint} as an equality constraint in accordance with the definition in \eqref{eq:cond_constraint}. The structure of the C/GMRES algorithm requires an equality constraint, so this step is necessary in order to avoid modifying the algorithm, making it applicable to  broader range of control problems.

\subsection{Cost function}\label{sc:CostFunction}

The cost function $J_{long}$ defined for the longitudinal control of the truck is given by
\begin{align}\label{eq:CostFunctionL}
    J_{long}(x,u) = J_{1} + J_{2} + J_{3} + J_{4} + J_{5} - J_{slk}.
\end{align}

In order to optimize the truck's speed in up-down slope sections and to maintain a reference value, the terms $J_1$ and $J_2$ are defined as follows
\begin{align}\label{eq:Cost1}
    J_{1} = w_{1}\left[ \frac{1}{2} \big( a(t) + g \cdot sin(\beta(s)) \big)^{2} \right]
\end{align}
\begin{align}\label{eq:Cost2}
    J_{2} = w_{2}\left[ \frac{1}{2} \left( v(t) - v_{Ref} \right)^{2} \right].
\end{align}

It is also important to penalize high accelerations during curves to improve safety, thus, $J_3$ is defined as
\begin{align}\label{eq:Cost4}
    J_{3} = e^{w_{3}(v^{2}(t) f_{curv}(s) - a_{lat_{max}})}.
\end{align}
Another important safety measure is to ensure the truck’s speed is under the path limit. This is possible by penalizing when $v$ is above ($V_{lim}$) as follows
\begin{align}\label{eq:Cost5}
    J_{4} = e^{w_{4}(v(t) - v_{lim})}.
\end{align}

Lastly, \eqref{eq:Cost6} is defined to reduce the fuel consumption at a steady condition, as previously shown by \cite{Caldas2019}, as follows
\begin{align}\label{eq:Cost6}
    J_{5} = e^{ w_{5} P_{eng}(t)}.
\end{align}

The cost function is subtracted by 
\begin{align}\label{eq:CostSlk}
    J_{slk} = w_{slk}u_{slk}
\end{align}
to avoid singularities when calculating the control input. Since the sign of $u_{slk}$ does not affect the optimality condition in \eqref{eq:EqConstraint} because it is squared, the solution of the optimization problem can bifurcate and the C/GMRES algorithm cannot determine the update of $u_{slk}$, which occurs when $u_{slk} = 0$. Thus, by multiplying $u_{slk}$ by a small positive constant $w_{slk}$, the value of $u_{slk}$ must be nonzero in order to meet the optimality condition defined in \eqref{eq:cond_hu} ($H_{u_{slk}} = \mu u_{slk} - w_{slk} = 0$, where $\mu$ is the Lagrange multiplier). This approach can be viewed as an interior penalty method, as detailed by \cite{gros_numerical}, and it is a necessary step after using the Auxiliary Variable Method described in Section \ref{subsec:constraints} to convert inequality constraints into equality constraints. As long as the optimality conditions and the constraints are met, the algorithm can find a solution to the optimization problem, where the complexity of the computation is related to the selection of the interior penalty weight $w_{slk}$ \citep{huang_nonlinear_2015}. The selection of initial guesses of the control input $u$ and $\mu$, called warm start, is also vital to obtain an accurate solution to the optimization problem \citep{ohtsuka2004}.

The weighting constants $w_{i}, i=1\ldots5$ were empirically adjusted after numerous simulations experiments presented by \cite{Caldas2019}, and are used to correctly balance the magnitude cost of each term of the cost function.

%%===================================================%%
%%                    Lat Control     %%
%%===================================================%%
\section{Lateral control strategy}
\label{sec:lat_ctrl}

To make the vehicle follow a desired given path, it is necessary to minimize the lateral displacement and heading error as well as ensure its stability. In order to do this, a system model that takes into account both path following and dynamic variables was adopted. This model was then used in a Robust Linear Quadratic Regulator (RLQR) based approach, generating an optimal feedback law. Finally, the generated feedback law was applied on the physical steering system. Furthermore, this task requires information on position, heading and velocity, which was provided by the GNSS and CAN bus systems. More details on the system architecture were given in \autoref{sec:architecture}.

This section briefly describes the lateral model and the robust regulator that was used in the vehicle's lateral control.

\subsection{Path-following model}
\label{path-following-model}

Solving the path-following problem consists of reducing both the lateral displacement and heading errors. To this end, a model based on the equations presented by \cite{Skjetne} was adopted. \autoref{path-following} illustrates the path-following problem for the heavy-duty vehicle. The lateral displacement is defined as the distance $\rho$ from the vehicle's center of gravity to the closed point $D$ on the desired path. The heading error is defined as $\theta = \psi_{1}-\psi_{des}$, where $\psi_{1}$ and  $\psi_{des}$ are respectively the current and desired angles.
%--
\begin{figure}[H]
    \centering
    \label{path-following-single}
    \includegraphics[width=0.65\linewidth]{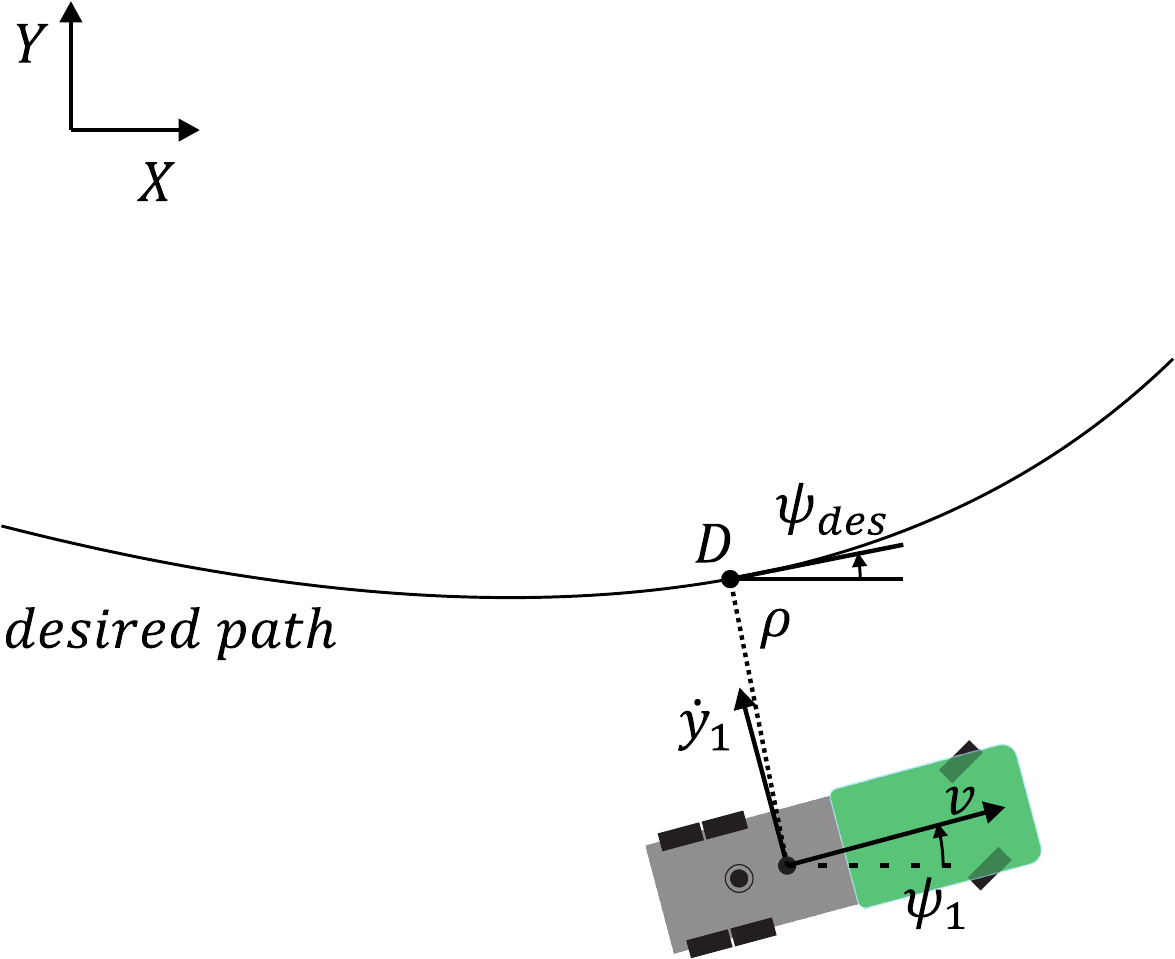}
    \caption{Path-following problem}
    \label{path-following}
\end{figure}
%--

Thus, the path-following model of the vehicle is expressed as
%--
\begin{equation}
\label{path-following-equation-1}
\begin{aligned}
\dot{\rho} &= v\sin\theta+\dot{y}_{1}\cos\theta\\
\dot{\theta} &= \dot{\psi},
\end{aligned}
\end{equation}
%--
where $\dot{y}_{1}$ and $v$ are respectively the vehicle's lateral and forward velocity. Then, by assuming a small heading error $\theta$, the lateral displacement $\rho$ can be rewritten in the linear form as
%--
\begin{equation}
\label{path-following-equation-2}
\dot{\rho} = v\theta+\dot{y}_{1}.
\end{equation}
%--

\subsection{The vehicle lateral model}
\label{lateral_model}

A single-track model was used to describe the dynamics of the vehicle's planar movement, neglecting roll and pitch effects. In other words, the vehicle is described by an equivalent track in each axle, linked by its wheelbase. Thus, the nonholonomic linear model presented by \cite{vandeMolengraftLuijten2012} was used in this work.

\autoref{single} shows the free body diagram of the heavy-duty vehicle, and \autoref{vehicles-parameters} details its parameters. Furthemore, the following assumptions are made:
%--
\begin{itemize}
\item The differences between left and right track are neglected;
\item The mass is assumed to be concentrated at the center of gravity;
\item Lateral tire forces are proportional to the tire slip angles;
\item There is no load transfer.
\end{itemize}
%--
%--
\begin{figure}[h]
	 \begin{center}
	    \includegraphics[width=0.5\linewidth]{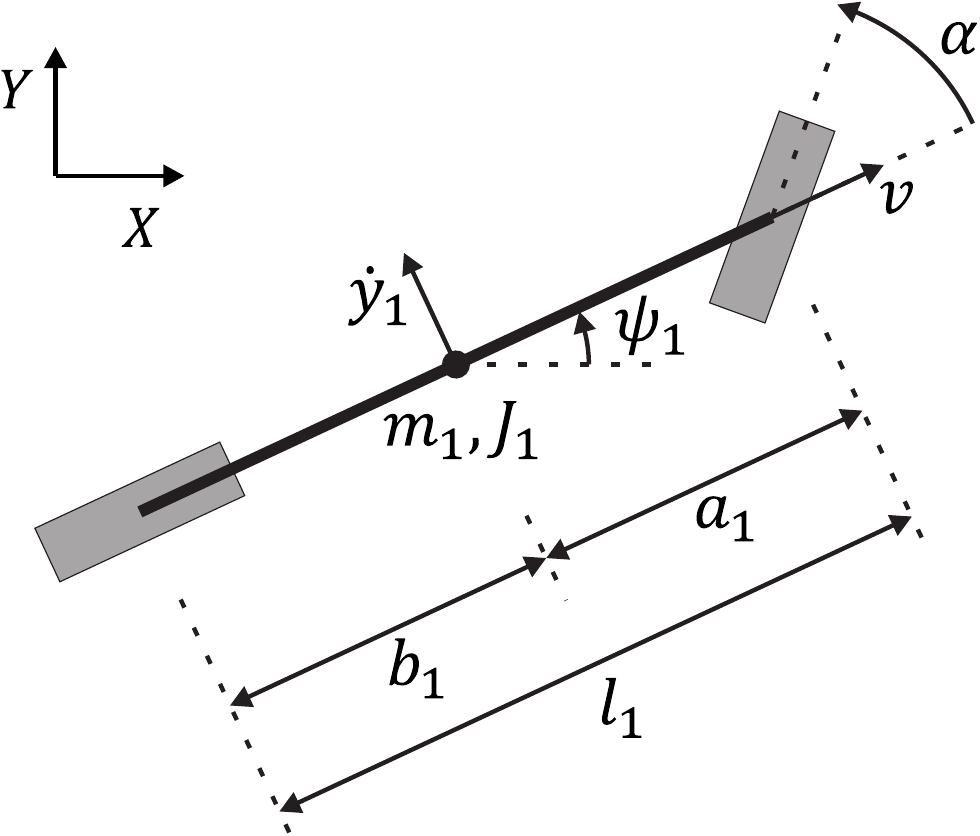}
	    \end{center}
	    \caption{Single-track model for a non-articulated vehicle.}
	    \label{single}
\end{figure}
\begin{table}[t]
\caption{Description of the vehicle parameters}
\label{vehicles-parameters}
\begin{center}
\resizebox{\columnwidth}{!}{
\begin{tabular}{c l c}
\hline
Parameter & Meaning & Unit\\
\hline
$a_{1}$ & Distance from the front axle to the CG on first unit & $m$\\
$b_{1}$ & Distance from the rear axle to the CG on first unit & $m$\\
$l_{1}$ & First unit wheelbase & $m$\\
$v$ & Forward velocity & $m/s$\\
$\dot{y}_{1}$ & Lateral Velocity & $m/s$\\
$m_{1}$ & First unit mass & $kg$\\
$J_{1}$ & First unit moment of inertia & $kg \hspace{1mm} m^{2}$\\
$\psi_{1}$ & First unit yaw & $rad$\\
$\alpha$ & Steering angle & $rad$\\
\hline
\end{tabular}}
\end{center}
\end{table}
%--
In addition to eliminate the lateral displacement  $\rho$ and heading error $\theta$, the vehicle's stability must be guaranteed. Hence, the lateral velocity $\dot{y}_{1}$ and yaw rate $\psi_{1}$ need to be controlled.

To this end, the equations of motion of the heavy-duty vehicle are expressed as
%--
\begin{equation}
    \label{eq:estate-stape}
    M\dot{x}(t)=A(t)x(t)+B\alpha(t),
\end{equation}
%--
with
%--
\begin{equation}
\label{eq:ssdynamic1}
\begin{aligned}
M=&\begin{bmatrix}
m_{1} & 0 & 0 & 0\\
0 & J_{1} & 0 & 0\\
0 & 0 & 1 & 0\\
0 & 0 & 0 & 1
\end{bmatrix}, \,\\
A(t)&=
\begin{bmatrix}
\frac{-c_{1}-c_{2}}{v(t)} & \frac{b_{1}c_{1}-a_{1}c_{1}-m_{1}v^{2}(t)}{v(t)} & 0 & 0\\
\frac{b_{1}c_{2}-a_{1}c_{1}}{v(t)} & \frac{-a_{1}^{2}c_{1}-b_{1}^{2}c_{1}}{v(t)} & 0 & 0\\
1 & 0 & 0 & v(t)\\
0 & 1 & 0 & 0
\end{bmatrix} \text{ and} \,
B=
\begin{bmatrix}
c_{1}\\
a_{1}c_{1}\\
0\\
0
\end{bmatrix},
\end{aligned}
\end{equation}
%--
where $M$ is the inertia matrix, $A(t)$ is the state matrix, $B$ is the input matrix and $\alpha$ is the control input, here, the steering angle. The state vector is defined as $x=[\dot{y}_{1},\dot{\psi}_{1}, \rho,\theta]^T$,
and $c_{1}$ and $c_{2}$ are the cornering stiffness of the front and rear axes of the vehicle, respectively.

As cited by \cite{luijten2010lateral}, \cite{Fancher1989} demonstrated that the relationship between the vertical load forces and tire cornering stiffness are approximately linear for heavy-duty vehicles. Therefore, $c_{j}$ escalates linearly with the vertical load force of the axle $F_{z_{j}}$ according to
%--
\begin{equation}
\label{cornering-stiffnes}
c_{j} = f_{j}F_{z_{j}} \text{ with } j = 1,\hdots,p,
\end{equation}
%--
where a normalized cornering stiffness $f_{i}$ is used and $p$ corresponds to the number of axles represented in the vehicle model, e.g. $j = 1$ corresponds to the front axle and $j = 2$ to the rear axle. Thus, the vertical force in each axle is calculated as
%--
\begin{equation}
\label{vertical-force}
\begin{aligned}
F_{z_{1}} &= m_{1}g\frac{b_{1}}{l_{1}},\\
F_{z_{2}} &= m_{1}g\frac{a_{1}}{l_{1}},
\end{aligned}
\end{equation}
where $g$ is the gravitational acceleration. Furthermore, \cite{houben2008analysis} observed that the normalized cornering stiffness of both tracks is approximately the same, thus $f_{1} \approx f_{2}$.

Finally, to be able to use a RLQR control approach, the state-space equation in (\ref{eq:estate-stape}) must be put in a suitable form. Hence, it is rewritten as
%--
\begin{equation}
    \label{eq:estate-space_rewritten}
    \dot{x}(t)=F(t)x(t)+G\alpha(t),
\end{equation}
where 
\begin{equation*}
    F(t) = M^{-1}A(t) \text{ and } G = M^{-1}B.
\end{equation*}
%--

With this in mind, the model in (\ref{eq:estate-space_rewritten}) is discretized and the RLQR is used as described next.

\subsection{The robust lateral controller}
\label{regulator}

The vehicle's lateral control was done by using the RLQR approach presented by \cite{Terra2014} and \cite{Cerri2009}. As aforementioned, the aim is to make the heavy-duty vehicle follow a desired path in the presence of parametric uncertainties. Thus, a feedback gain $K_k$ is computed, giving the control inputs in the form $\alpha_k=K_k x_k$. These control inputs minimize a quadratic cost function while keeping stability in the presence of parametric uncertainties. This way, the optimization problem can be formulated as 
%--
\begin{equation}
\label{eq:minmaxcost}
\underset{x_{k+1},\alpha_{k}}{min} \ \underset{{\delta} F,{\delta} G}{max} {\bar{J}^{\mu}_{k}(x_{k+1},\alpha_{k},{\delta} F,{\delta} G)},
\end{equation}
%--
where the uncertainties in the state and input matrices are modeled as
%--
\begin{equation}
\label{eq:uncertainties}
\begin{bmatrix}
\delta F & \delta G
\end{bmatrix} = H \Delta \begin{bmatrix}
E_{F} & E_{G}
\end{bmatrix},
\end{equation}
%--
with $\Delta$ an arbitrary contraction matrix such that $||\Delta|| \leq 1$, and $H$, $E_{F}$ and $E_{G}$ are known matrices obtained in similar fashion as presented by \cite{Barbosa2019}. See \autoref{sec:uncertainties}. Furthermore, \cite{deMorais2022} presented a method to obtain these matrices via multiobjective optimization, which could be an alternative to the analytical method used here.

Then, using the model obtained from (\ref{eq:estate-stape}), (\ref{eq:ssdynamic1}) and (\ref{eq:estate-space_rewritten}), and the  criterion (\ref{eq:minmaxcost}), the vehicle's lateral controller is implemented as the RLQR in \autoref{rlqr_algor}. Just as in the standard LQR formulation, $Q \succ 0$ and $R \succ 0$ are respectively the weighting matrices for the states and inputs, and $P \succ 0$ the solution of the associated algebraic Riccati equation. The closed loop matrix $L$, the feedback gain $K$ and $P$ are obtained from (\ref{eq:framework})  in \autoref{rlqr_algor}, and the control inputs are subsequently obtained from the feedback law. The algorithm is executed until the vehicle reaches the last waypoint $\chi_f$. A brief description of the RLQR is given in the \autoref{RLQR_description} and we refer to the papers of \cite{Terra2014} and \cite{Cerri2009} for more details.

%---------------------------------------------------------------------------------------------------------------------------------------
\begin{algorithm*}[h]
        \begin{algorithmic}[1]
        \Require  $x_{0}$ and $P_{0}\succ{0}$
        \While{$\chi_k \neq \chi_f$}
           \State \begin{equation}
           \label{eq:framework}
               \begin{bmatrix}
	           L_{k}\\
	           K_{k}\\
	           P_{k+1}
	\end{bmatrix}=
	\begin{bmatrix}
	0 & 0 & 0 \\
	0 & 0 & 0 \\
	0 & 0 & -I \\
	0 & 0 & F_{k} \\
	0 & 0 & E_{F} \\
	I & 0 & 0 \\
	0 & I & 0
	\end{bmatrix}^{T}
	\begin{bmatrix}
	P_{k}^{-1}& 0 & 0 & 0 & 0 & I & 0\\
	0 & R^{-1} & 0 & 0 & 0 & 0 & I\\
	0 & 0 & Q^{-1} & 0 & 0 & 0 & 0\\
	0 & 0 & 0 & 0 & 0 & I & -G\\
	0 & 0 & 0 & 0 & 0 & 0 & -E_{G}\\
	I & 0 & 0 & I & 0 & 0 & 0\\
	0 & I & 0 & -G^{T} & -E_{G}^{T} & 0 & 0
	\end{bmatrix}^{-1}
	\begin{bmatrix}
	0\\
	0\\
	-I\\
	F_{k}\\
	E_{F}\\
	0\\
	0
	\end{bmatrix}
        \end{equation}
        \State \begin{equation*}
        \alpha_{k}	= K_{k}x_{k}.
        \end{equation*}
        \EndWhile
        \end{algorithmic}
	\caption{The implemented Robust Linear Quadratic Regulator}
	\label{rlqr_algor}
\end{algorithm*}

%---------------------------------------------------------------------------------------------------------------------------------------

\begin{remark}
    Note that there is a major difference between \autoref{rlqr_algor} and \autoref{rlqr_algor_ap} (in \autoref{RLQR_description}). In the latter, all the closed loop matrices, feedback gains and the solution of the Riccati equation are obtained in a \textit{backward} step with $k=N-1,\ldots,0$, where $N$ is the horizon. These results are then used to obtain the optimal control inputs $u_{k}^{*}$ and the regularized state vectors $x_{k+1}^{*}$ in a \textit{forward} step with $k=0,...,N-1$. This procedure is however not suitable for our application since the $N$ is unknown a priori. Thus, in the former a \textit{only-forward} procedure is used and the control inputs $\alpha$ are calculated until the vehicle arrives at $\chi_f$. This is possible because the regulator converges much faster than the variations in $A$ are perceived.
\end{remark}

%%===================================================%%
%%                    Results      %%
%%===================================================%%
\section{Experiments and results}
\label{sec:exp_res}

This section first brings the parameters used for both longitudinal and lateral control designs, and subsequently, shows the results obtained in a mining environment.

\subsection{Vehicle parameters}

\autoref{tab:parameters-values-non-articulated} shows the vehicle's parameters, necessary for the lateral and longitudinal controller design. This information may be obtained from technical specifications of cargo transportation vehicles\footnote{https://www.scania.com} and tipper implement\footnote{http://www.librelato.com.br} manufacturers. Similarly to what was done by \cite{Barbosa2019}, the normalized cornering stiffness was calculated assuming $f_{1} = f_{2} = 5.73\hspace{1mm}rad^{-1}$. 

%--
\begin{table}[htb]
\caption{Description of the non-articulated vehicle parameters}
\label{tab:parameters-values-non-articulated}
\begin{center}
\resizebox{\columnwidth}{!}{%
{\begin{tabular}{l l l}
\hline
Description & Parameter & Value\\
\hline
Distance from the front axle to the CG & $a_{1}$ & $3.19\hspace{1mm}m$\\
Distance from the rear axle to the CG & $b_{1}$ & $1.62\hspace{1mm}m$\\
Vehicle wheelbase & $l_{1}$ & $4.81\hspace{1mm}m$\\
Vehicle mass & $m_{1}$ & $16030\hspace{1mm}kg$\\
Payload mass  & $m_{2}$ & $35000\hspace{1mm}kg$\\
Vehicle moment of inertia & $J_{1}$ & $215717\hspace{1mm}kg \hspace{1mm} m^{2}$\\
Front axle cornering stiffness & $c_{1}$ & $540419\hspace{1mm}N/rad$ \\
Rear axle cornering stiffness & $c_{2}$ & $1064462\hspace{1mm}N/rad$\\
\multirow{12}{*}{Gear ratios} & $\xi_1$  & 11.32 \\
                              & $\xi_2$  & 9.164 \\
                              & $\xi_3$  & 7.194 \\
                              & $\xi_4$  & 5.823 \\
                              & $\xi_5$  & 4.632 \\
                              & $\xi_6$  & 3.750 \\
                              & $\xi_7$  & 3.019 \\
                              & $\xi_8$  & 2.444 \\
                              & $\xi_9$  & 1.918 \\
                              & $\xi_{10}$ & 1.553 \\
                              & $\xi_{11}$ & 1.235 \\
                              & $\xi_{12}$ & 1.000 \\
\hline
\end{tabular}}
}
\end{center}
\end{table}
%--

\subsection{Control parameters}

The NMPC longitudinal velocity planner parameters are presented in \autoref{table:NMPCPar_exp}, adjusted according to the vehicle's current weight and surface condition of the route, considering an initial value for reference speed is set as 40 km/h.
%-
\begin{table}[ht]
\small
\caption{NMPC parameters used in the experiment.}%
\label{table:NMPCPar_exp}
\begin{center}
\resizebox{\columnwidth}{!}{%https://www.overleaf.com/project/5d4db4349399f14c2b0c998b
{\begin{tabular}{l  l  l}
\hline
%\toprule
\textbf{Description}&   \textbf{Parameter} & \textbf{Value} \\
\hline
		Prediction horizon           & $T$            & 10s     \\
		C/GMRES iterations  & $kmax$          & 10         \\ 		
		Control input constraint & $u_{max}$ & [-4.5,4.5]N/kg \\
		Lateral acceleration constraint             &$a_{lat}$             & [-0.5,0.5]$m/s^{2}$        \\
		Reference speed    &$v_{Ref}$         & 40km/h  \\
        Route adhesion coefficient &$\gamma$ &0.3 \\
        Fixed distance for slope calculation & $\Delta s$ & 20m \\
		Cost function weights    &$w_{1}$        & 30         \\
		                                          &$w_{2}$         & 2.5     \\
		                                          &$w_{3}$         & 1 \\
		                                          &$w_{4}$        & 1 \\
		                                          &$w_{5}$        & 2.5 \\
		                                          &$w_{slk}$        & 0.25 \\
\hline
\end{tabular}}
}
\end{center}
\end{table}
%--

Regarding the lateral control, the parameters used for the RLQR in the experiments were
%--
\begin{gather*}
\mu = 10^{9},\\
H =
\begin{bmatrix}
1\\
1\\
1\\
1
\end{bmatrix},\hspace{1mm}E_{A} = 
\begin{bmatrix}
-0.000405618009134\\
 0.004949413574869\\
 0.000000034165611\\
 0.000024747067874
\end{bmatrix}^{T}\textnormal{and} \\ \hspace{1mm}E_{B} =
\begin{bmatrix}
-0.00114326083475285\\
-0.00114326083475285
\end{bmatrix}^{T},
\end{gather*}
%--
and the weight matrix was set to
%--
\begin{equation*}
Q = 
\begin{bmatrix}
0.1 & 0 & 0 & 0\\
0 & 0.1 & 0 & 0\\
0 & 0 & 100 & 0\\
0 & 0 & 0 & 15
\end{bmatrix} \text{and} \hspace{1mm}
R = 
\begin{bmatrix}
10000 & 0\\
0 & 10000
\end{bmatrix},
\end{equation*}
%--
with $H$, $E_{A}$ and $E_{B}$ obtained using the method in \autoref{sec:uncertainties}.

The state matrix $A$ and input matrix $B$ are obtained with the parameters from \autoref{tab:parameters-values-non-articulated}. Note that the velocity $v$ varies along the trajectory, implying a also varying $A$ and constituting a linear time-variant system. Moreover, since some terms of $A$ are divided by $v$, we set $v=1.38m/s$ as the minimum allowed velocity to avoid division by zero.

%\rev{Inserir comentário sobre as matrizes que permanecem constantes, e sobre o calculo das matrizes $A_i$ e $B_i$}.

In the lateral model, the payload parameter was considered to be $12550\hspace{1mm}kg$, which is a suitable value for regional and local freight transportation, construction and some off-road applications. However, the usual value for payload in iron ore mining operations is in the order of 35 tons. This was intentional, so that we could investigate the effects of the parametric uncertainties on the performance of the controller. Thus, the RLQR computed a feedback gain based on a payload value that is much lower than the actual value. Furthermore, due to safety recommendations, the maximum steering angle was limited to $\norm{\alpha} = 0.3491\hspace{1mm}rad$.

For these experiments, the orientation and lateral velocities were not measured by the PPP GNSS. Consequently, only the lateral displacement $\rho$ and the orientation error $\theta$ were measured, despite the control gain being computed considering all state variables. However, $\dot{y}_{1}$ and $\dot{\psi}_{1}$ have very small values for the performed path. Thus, this limitation was circumvented by setting them to zero without major impact on the results.

\subsection{Results}

Experimental tests were conducted on an unpaved road in the Carajas' mine. As shown in \autoref{fig:exp_route}, the map can be divided into two paths: First, starting at A towards B, then from B towards A. However, due to safety recommendations, only the B-A path was considered in the experiments.
\begin{figure}[H]
    \centering
    \includegraphics[width=0.5\linewidth]{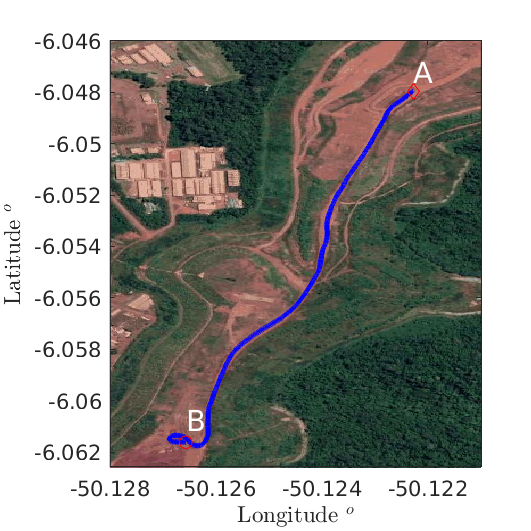}
    \caption{Path B-A used in the experiment in the Carajas' mine.}
    \label{fig:exp_route}
\end{figure}

\autoref{fig:trackBA-long} and \autoref{fig:lateral_BA} show the results obtained on path B-A for the longitudinal and lateral controls, respectively. The reference velocity was set around 20km/h, which was calculated based on the vehicle's parameters and imposed constraints, and a $max\|\dot{\alpha}\|= 0.5057\hspace{1mm}rad/s$ was observed.
%--
\begin{figure}[H]
    \centering
    \begin{subfigure}{.75\textwidth}
      \centering
      \includegraphics[width=\linewidth]{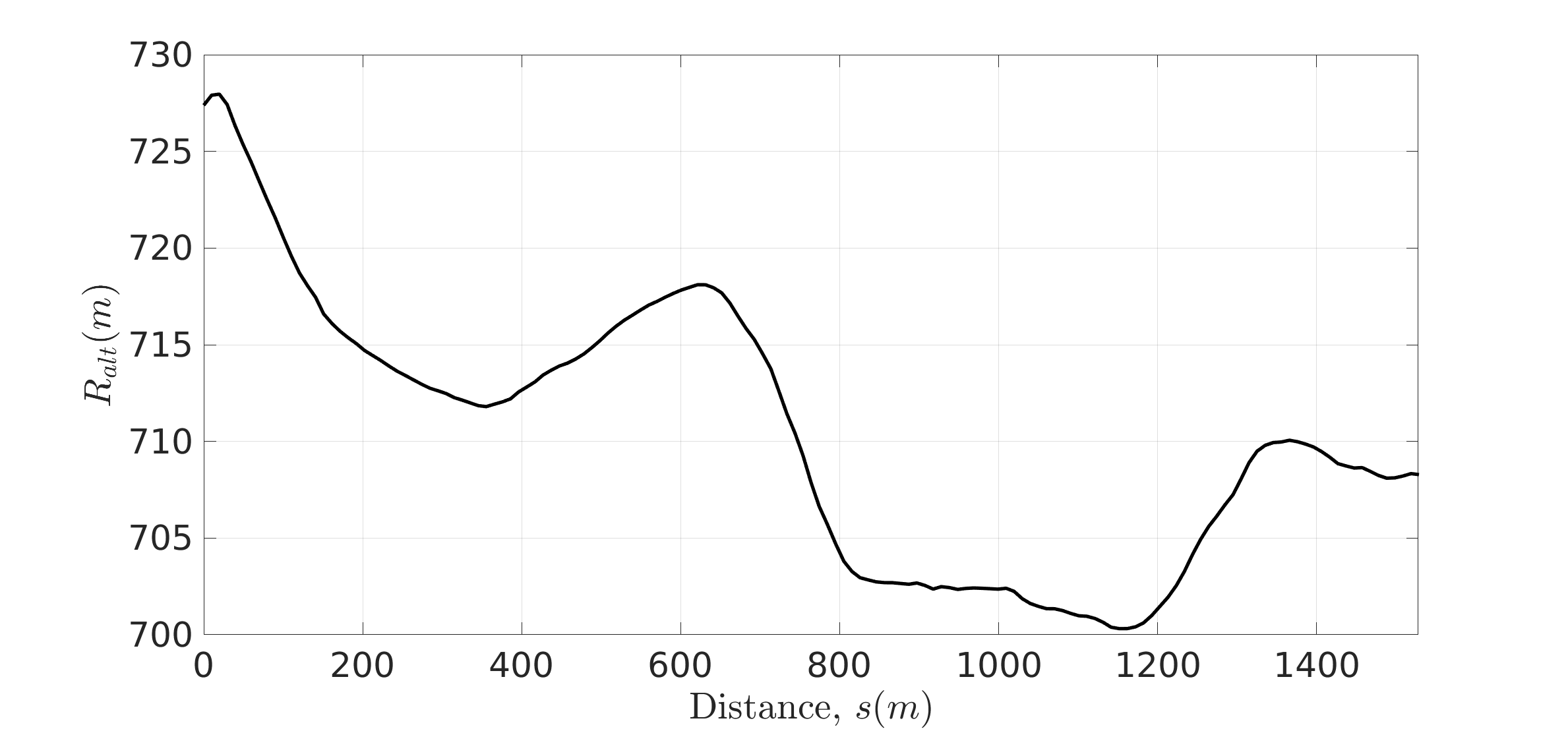}
      \caption{Altitude profile}
      \label{fig:altitude_BA}
    \end{subfigure}
    
    \begin{subfigure}{.75\textwidth}
      \centering
      \includegraphics[width=\linewidth]{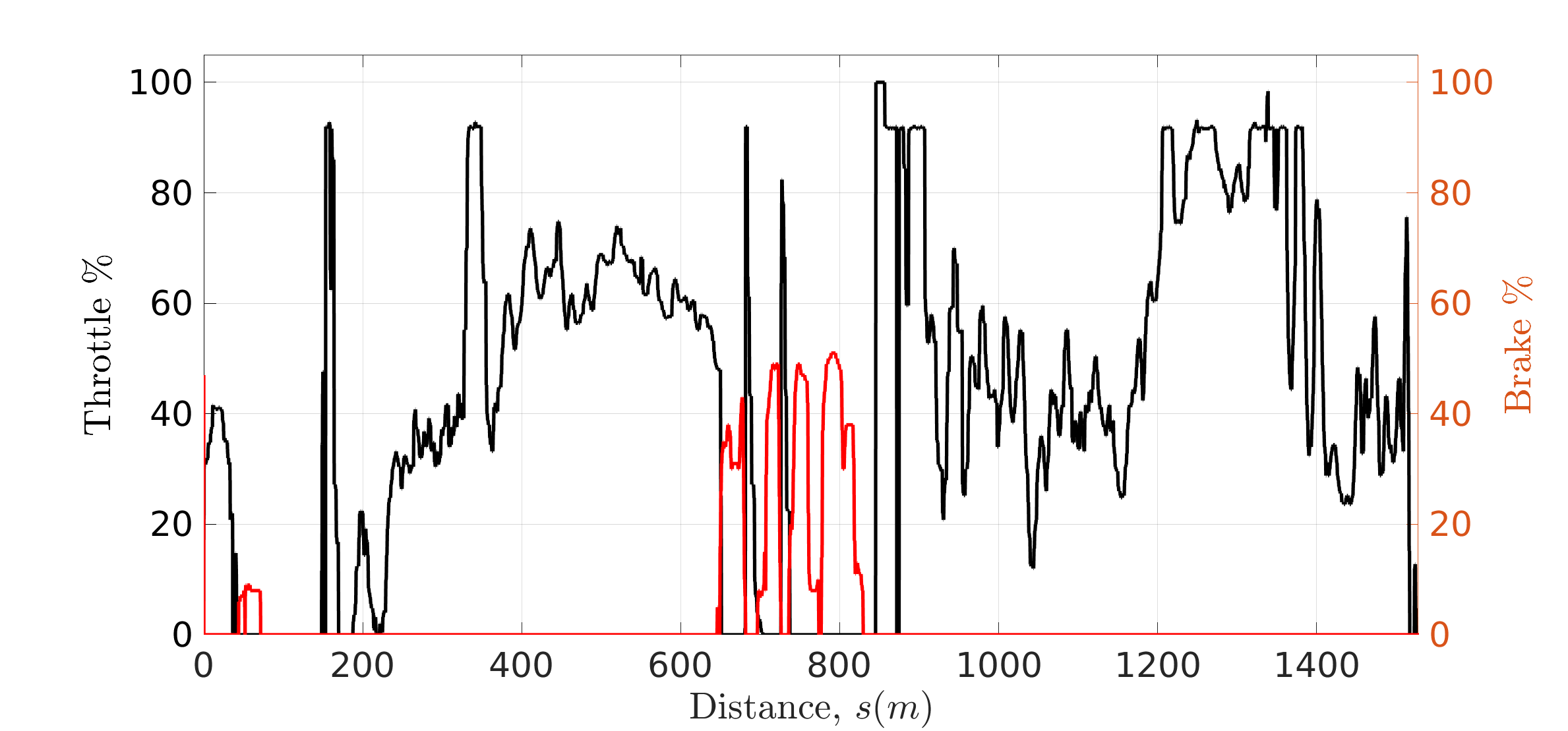}
      \caption{Throttle and brake inputs}
      \label{fig:long_control_BA}
    \end{subfigure}
    
    \begin{subfigure}{.75\textwidth}
      \centering
      \includegraphics[width=\linewidth]{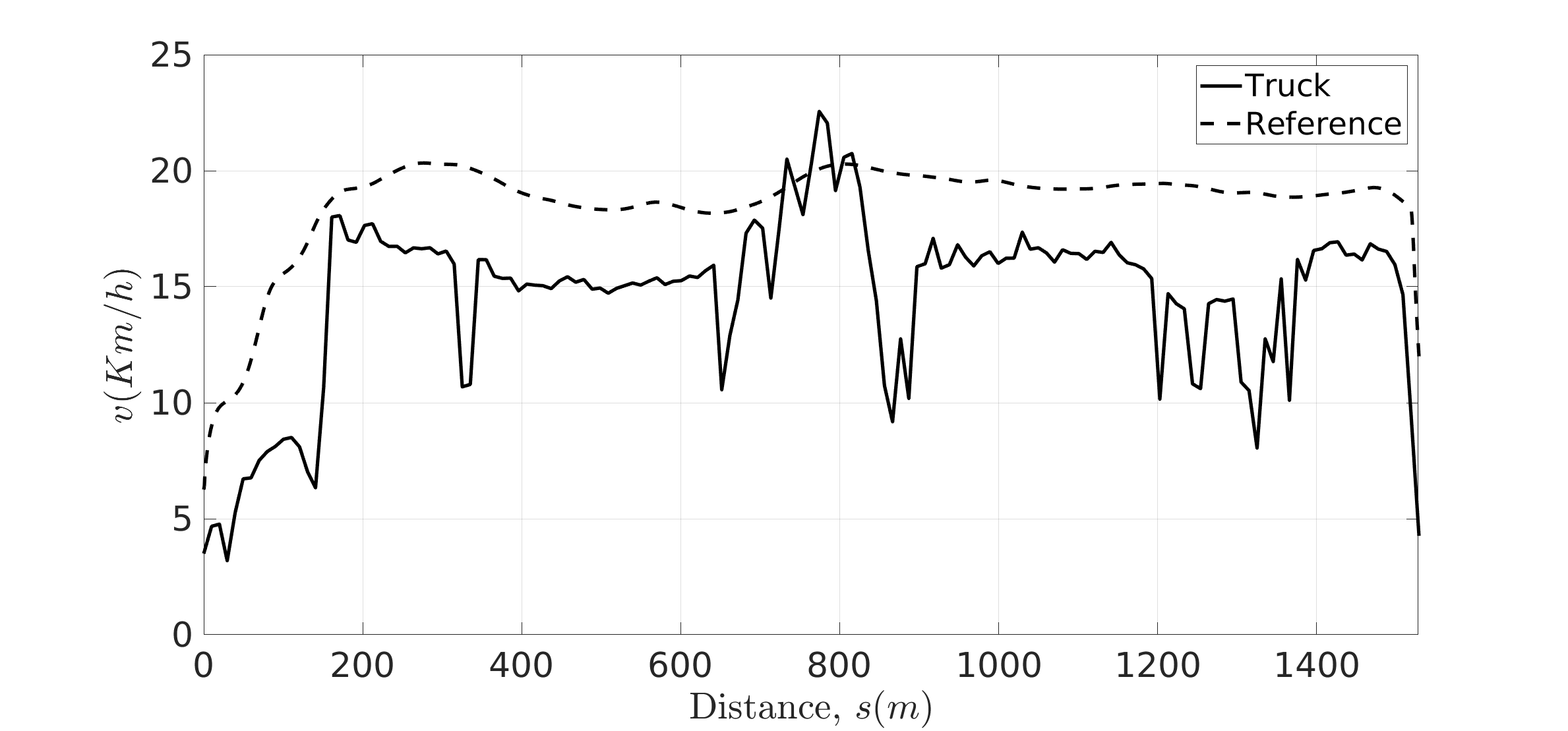}
      \caption{Velocity}
        \label{fig:velocity_BA}
    \end{subfigure}
    \caption{Experiment results with the longitudinal NMPC velocity planner and PI controller in path B-A.}
    \label{fig:trackBA-long}
\end{figure}
%--
\begin{figure}[H]
    \centering
    \begin{subfigure}{.75\textwidth}
      \centering
      \includegraphics[width=\linewidth]{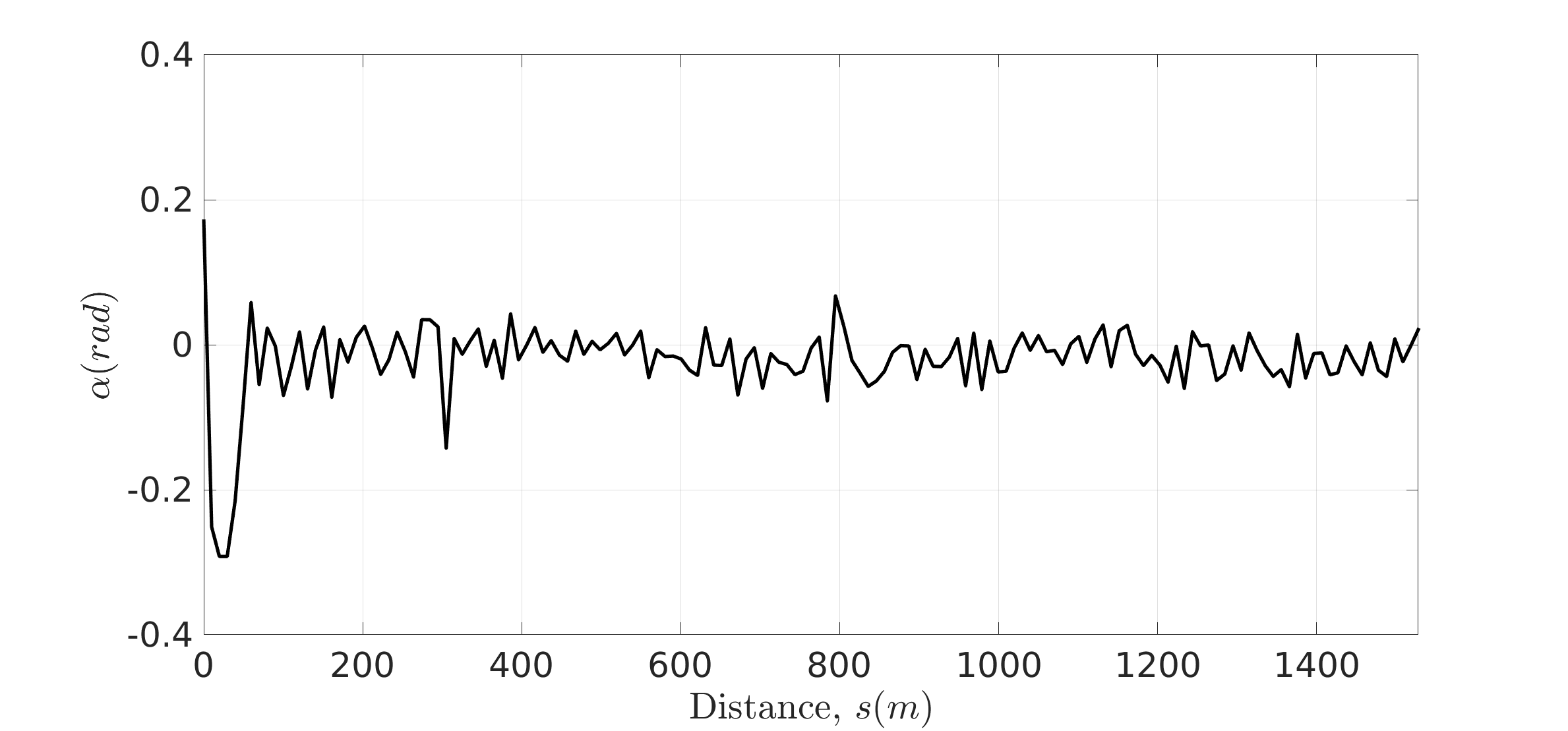}
      \caption{Lateral control}
      \label{fig:lat_control_BA}
    \end{subfigure}
    
    \begin{subfigure}{.75\textwidth}
      \centering
      \includegraphics[width=\linewidth]{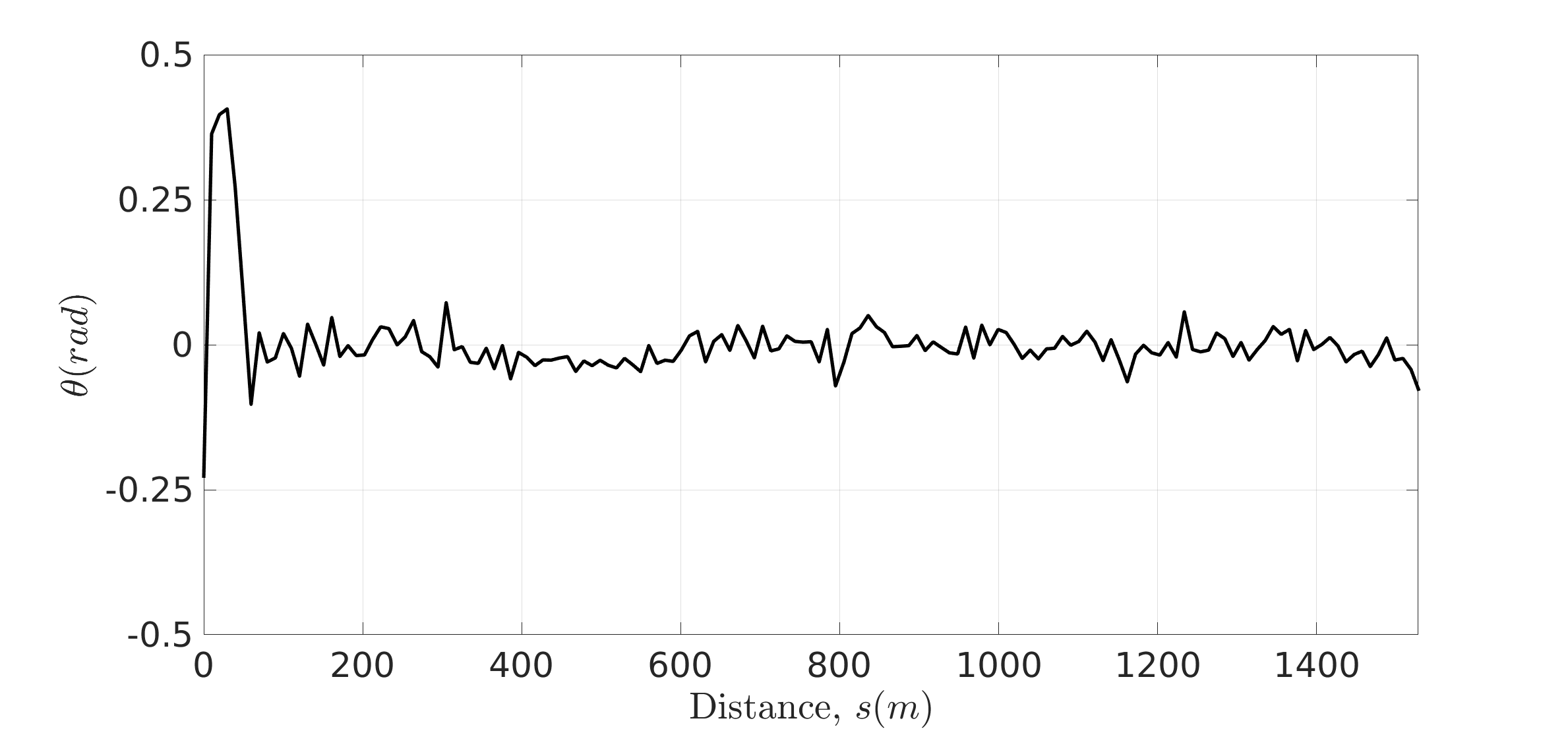}
      \caption{Heading error.}
        \label{fig:heading_BA}
    \end{subfigure}
    
    \begin{subfigure}{.75\textwidth}
      \centering
      \includegraphics[width=\linewidth]{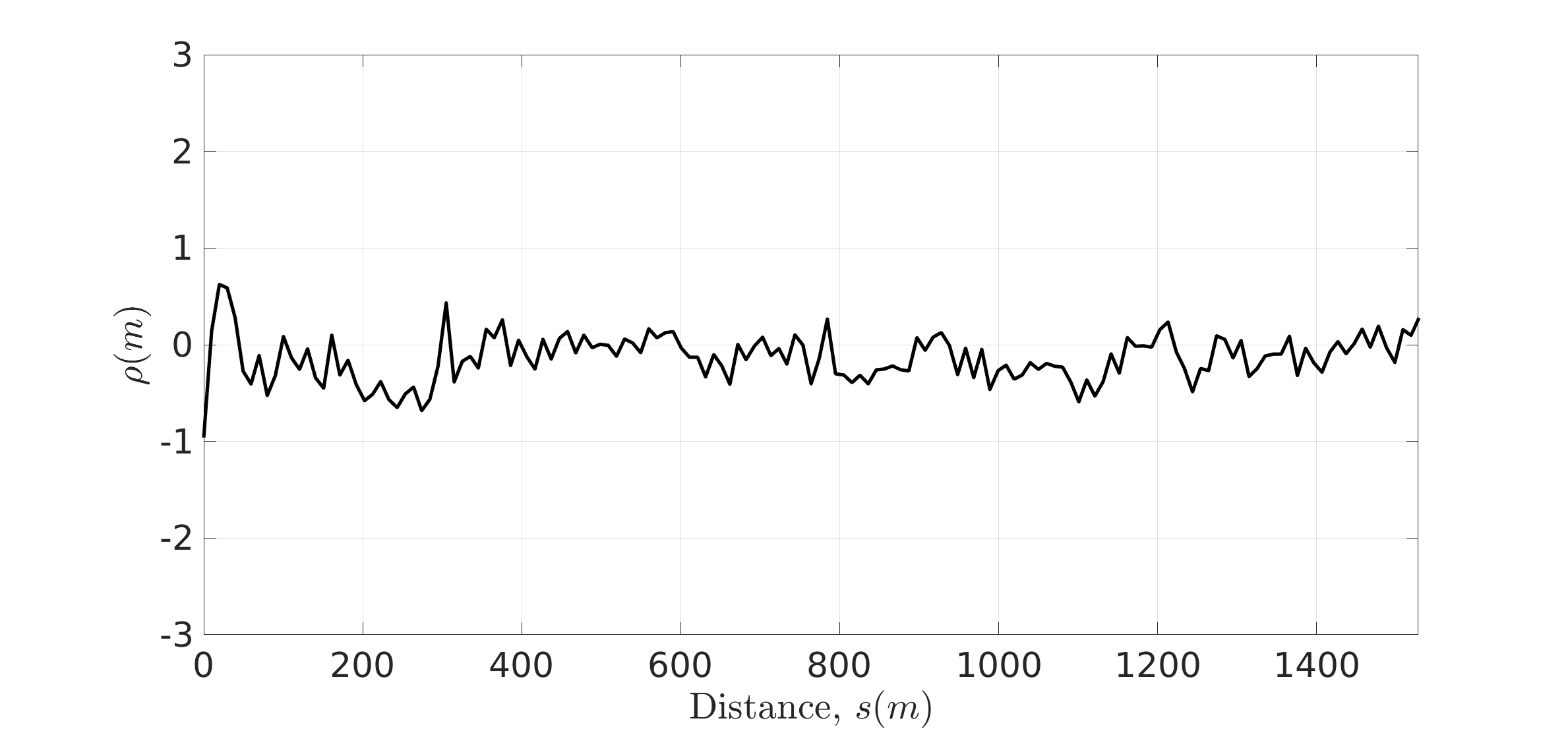}
      \caption{Lateral displacement}
      \label{fig:lat_error_BA}
    \end{subfigure}
    \caption{Experiment results with the RLQR lateral controller in track B-A.}
    \label{fig:lateral_BA}
\end{figure}
%--

\section{Discussion}
\label{sec:discussion}

The main goal of these experiments was to investigate the performance of the proposed autonomous truck driving system in an off-road environment, evaluated in an open-pit mining scenario. It is worth noting that comparisons with other controllers were not carried out in this paper, since a detailed evaluation of the proposed methods was already addressed on simulations in our previous works \citep{Barbosa2019,Caldas2019}.

The longitudinal NMPC velocity planner, in combination with a PI controller,  was able to smoothly adjust the truck's velocity, operating under the reference value and gradually decelerating when reaching the target goal. Before uphills, the controller accelerates the vehicle to reduce the engine's effort at critical slope sections. At downhills, it uses the gravitational force to maintain the vehicle's speed close to the reference value with reduced throttle levels. 

Regarding the lateral control, the RLQR presented satisfactory robustness and driving smoothness in the presence of large mass variation, a distinctive characteristic of mining operations. Driving smoothness can be verified by the $max\|\dot{\alpha}\|= 0.5057\hspace{1mm}rad/s$ obtained on path B-A. Additionally, $\dot{y}_1$ and $\dot{\psi}_1$ were very small for the paths performed in the experiments. Thus, it is likely that the limitations in feedback data did not significantly affect the path-following task.

%%===================================================%%
%%                    Conclusion     %%
%%===================================================%%

\section{Conclusion}
\label{sec:conclusion}

Experimental results of longitudinal and lateral path tracking control of an autonomous truck in a real mining environment have been presented in this paper.

The performance of the proposed system showed satisfactory results. The truck, operating in autonomous mode with its full load, kept within the imposed constraints along the entire path, showing robustness and driving smoothness. This means that parametric uncertainties were satisfactorily handled, and abrupt braking, steering and accelerations were avoided, ensuring a safer and more comfortable ride. 

Considering the promising results obtained in this paper, future work using the proposed autonomous truck can include different sensors such as cameras, LiDARs (Light Detection And Ranging) and radars, for example. Thus, it will be possible to include obstacle avoidance, fault detection and computer vision for tracking and detection features, which may improve off-road applications and even open a path for an urban scenario implementation. Additionally, methods that use machine learning to obtain the uncertainty matrices may be an alternative to the analytical method and improve performance of the RLQR controller.

\backmatter

%%===================================================%%
%%                    Statements and Declarations     %%
%%===================================================%%
\section*{Statements and Declarations}

\begin{itemize}
\item \textbf{Funding:}
This study was financed in part by Vale S.A., the Coordenação de Aperfeiçoamento de Pessoal de Nível Superior - Brasil (CAPES) - Finance Code 001 and 88887.136349/2017-00, the Brazilian National Research Council (CNPq) under grant 465755/2014-3, and the S\~ao Paulo Research Foundation (FAPESP) under grant 2014/50851-0.
\item \textbf{Conflict of interest/Competing interests:}
The authors have no competing interests to declare that are relevant to the content of this article.
\item \textbf{Ethics approval:} Not applicable.
\item \textbf{Consent to participate:} Not applicable.
\item \textbf{Consent for publication:} Not applicable.
\item \textbf{Availability of data and materials:} Not applicable.
\item \textbf{Code availability:} Not applicable.
\item \textbf{Authors' contributions:} 
    \begin{itemize}
        \item Kenny A. Q. Caldas: Conceptualization, Methodology, Validation, Formal Analysis, Visualization, Writing - Original Draft and Writing - Review \& Editing;
        \item Filipe M. Barbosa: Conceptualization, Methodology, Validation, Writing - Original Draft and Writing - Review \& Editing;
        \item Junior A. R. Silva: Software and Writing  - Original Draft;
        \item Tiago C. Santos: Software, Investigation;
        \item Iago P. Gomes: Software, Writing - Original Draft;
        \item Luis A. Rosero: Resources, Investigation;
        \item Denis F. Wolf: Conceptualization, Writing - Review \& Editing, Supervision, Funding acquisition;
        \item Valdir Grassi Jr: Conceptualization, Writing - Review \& Editing, Supervision, Funding acquisition.
    \end{itemize}
\end{itemize}

%%===================================================%%
%%         Appendices    %%
%%===================================================%%

\begin{appendices}

\section{Robust Linear Quadratic Regulator}
\label{RLQR_description}

\subsection{Problem formulation}
Consider a discrete-time system subject to parametric uncertainties as follows
\begin{equation}
\label{eq:statespaceDT}
x_{k+1} = (A_{k} + \delta A_{k})x_{k} + (B_{k}  + \delta B_{k})u_{k},
\end{equation}
where $k = 0,\hdots, N$, $x_{k}$ is the state vector, $u_{k}$ is the control input, and $A_{k}$ and $B_{k}$ are the nominal discrete model matrices. The unknown matrices $\delta A_{k}$ and $\delta B_{k}$ represent the uncertainties in the parameters and are modeled as
%--
\begin{equation}
\label{eq:RLQRuncertainties}
\begin{bmatrix}
\delta {A}_{k} & \delta B_{k}
\end{bmatrix} = H_{k} \Delta_{k} \begin{bmatrix}
E_{A_{k}} & E_{B_{k}}
\end{bmatrix},
\end{equation}
%--
where $k = 0,\hdots, N$, $H_{k}$, $E_{A_{k}}$ and $E_{B_{k}}$ are known matrices, and $\Delta_{k}$ is an arbitrary contraction matrix such that $||\Delta|| \leq 1$.

Then, the RLQR is obtained by solving the following optimization problem:
%--
\begin{equation}
\label{eq:rlqrminmax}
\underset{x_{k+1},u_{k}}{min} \ \underset{{\delta} A_{k},{\delta} B_{k}}{max} {\bar{J}^{\mu}_{k}(x_{k+1},u_{k},{\delta} A_{k},{\delta} B_{k})},
\end{equation}
%--
where
%--
\begin{align}
\label{eq:rqlqrcostfun}
 \bar{J}^{\mu}_{k}(x_{k+1},u_{k},{\delta} A_{k},{\delta} B_{k}) = 
 \begin{bmatrix}
x_{k+1}\\
u_{k}
\end{bmatrix}^T
\begin{bmatrix}
P^{r}_{k+1} & 0 \\
0 & R_{k}
\end{bmatrix} 
\begin{bmatrix}
x_{k+1}\\
u_{k}
\end{bmatrix} + \Phi^T 
\begin{bmatrix}
Q_{k} & 0\\
0 & \mu I
\end{bmatrix}\Phi
\end{align}
%--
is the cost function, with
%--
\begin{equation*}
\Phi = \left\{
\begin{bmatrix}
0 & 0 \\
I & -B_{k}-\delta B_{k}
\end{bmatrix}
\begin{bmatrix}
x_{k+1}\\
u_{k}
\end{bmatrix}
- \begin{bmatrix}
-I\\
A_{k}+{\delta} A_{k}
\end{bmatrix}
x_{k}\right\},
\end{equation*}
%--
the fixed penalty parameter $\mu > 0$, and the weighing matrices $Q_{k} \succ 0$, $R_{k} \succ 0$ and $P_{k+1} \succ 0$.

See the work of \cite{Cerri2009} for details on penalty function.

\subsection{The regulator}

The solution for the optimization problem (\ref{eq:rlqrminmax})-(\ref{eq:rqlqrcostfun}) is obtained based on the solution of a general robust regularized least-squares problem. Furthermore, the regulator has an optimal operating point for each step in the algorithm when $\mu > 0$. Thus, its regularization is reached as a result of the minimization over both $x_{k+1}(\mu)$ and $u_{k}(\mu)$ \citep{Terra2014}.

The following theorem provides a framework to compute the optimal state trajectory, control inputs and cost function. This is given in terms of an array of matrices \citep{Terra2014}.
%--
\begin{theorem}
\label{theorem_rlqr}
The optimal solution for (\ref{eq:rlqrminmax})-(\ref{eq:rqlqrcostfun}) is given by 
%--
\begin{equation} 
\label{exp_a}
\begin{bmatrix}
x^{\ast}_{k+1}(\mu)\\
u^{\ast}_{k}(\mu)\\
\tilde{J}^{\mu}_{k}(x^{\ast}_{k+1}(\mu),u^{\ast}_{k}(\mu))
\end{bmatrix}=\begin{bmatrix}
I & 0 & 0 \\
0 & I & 0 \\
0 & 0 & x_{k}(\mu)^{T}
\end{bmatrix}^{T}\begin{bmatrix}
L_{k,\mu}\\
K_{k,\mu}\\
P_{k,\mu}
\end{bmatrix}x_{k},
\end{equation}
%--
for each $\mu > 0$, where the closed-loop matrix $L_{k}$ and the feedback gain $K_{k}$ result from the recursion
\begin{equation}
\label{eq:rlqrsolution}
\begin{bmatrix}
L_{k} \\ K_{k} \\ P_{k}
\end{bmatrix} =
\begin{bmatrix}
0 & 0 & -I & \mathcal{A}_{k} & 0 & 0\\
0 & 0 & 0 & 0 & 0 & I\\ 
0 & 0 & 0 & 0 & I & 0
\end{bmatrix}
\Xi^{-1}
\begin{bmatrix}
0 \\ 0 \\ -I \\ \mathcal{A}_{k} \\ 0 \\ 0
\end{bmatrix},
\end{equation}
%--
with
%--
\begin{equation}
\Xi = \begin{bmatrix}
P^{-1}_{k+1} & 0 & 0 & 0 & I & 0 \\
0 & R^{-1}_{k} & 0 & 0 & 0 & I \\
0 & 0 & Q^{-1}_{k} & 0 & 0 & 0 \\
0 & 0 & 0 & \Sigma_{k}\left(\mu,\hat \lambda_{k}\right) & \mathcal{I} & -\mathcal{B}_{k} \\
I & 0 & 0 & \mathcal{I}^{T} & 0 & 0 \\
0 & I & 0 & -\mathcal{B}^{T} & 0 & 0
\end{bmatrix}, \nonumber
\end{equation}
\begin{align*}
& \Sigma_{k} = \begin{bmatrix}
\mu^{-1}I - \hat{\lambda}^{-1}_{k}H_{k}H_{k}^{T}  & 0 \\
0 & \hat{\lambda}^{-1}_{k}I
\end{bmatrix}, \nonumber \\
& \mathcal{I} = \begin{bmatrix}
I \\ 0
\end{bmatrix},\
\
\mathcal{A}_{k} = \begin{bmatrix}
A_{k} \\ E_{A_{k}}
\end{bmatrix}, \
\mathcal{B}_{k} = \begin{bmatrix}
B_{k} \\ E_{B_{k}}
\end{bmatrix},
\end{align*}
%--
\noindent where $P_{k+1}$ is the solution of the associated Riccati Equation and $\lambda_{k}> \|\mu H_{k}^TH_{k}\|$.
%--
Alternatively, we have
%--
\begin{equation}
\label{alternative}
\begin{aligned}
P_{k,\mu} = &L^{T}_{k,\mu}P_{k+1}L_{k,\mu} + K_{k,\mu}R_{k}K_{k,\mu} + Q_{k}+\\
&(\mathcal{I}L_{k,\mu} - \mathcal{B}_{k}K_{k,\mu} - \mathcal{A}_{k})^{T}\Sigma_{k,\mu}^{-1}(\mathcal{I}L_{k,\mu}-\mathcal{B}_{k}K_{k,\mu}-\mathcal{A}_{k}) \succ 0.
\end{aligned}
\end{equation}
%--
\end{theorem}
%--
\begin{proof}
It follows from the results shown by \cite{Cerri2009}.
\end{proof}
%--
The parameter $\mu$ relates to system's robustness, thus ensuring regularization and the validity of (\ref{eq:statespaceDT}). This means that for maximum robustness, $\mu \rightarrow \infty$ and consequently $\Sigma_{k} \rightarrow 0$.

The matrix $P_{k,\mu}$ is finite and $\mathcal{I}L_{k,\mu}-\mathcal{G}_{k}K_{k,\mu}-\mathcal{F}_{k} \rightarrow 0$ for each iteration of (\ref{alternative}), as shown by \cite{Terra2014}. Therefore,
%--
\begin{equation}
\label{feedback}
\begin{aligned}
&L_{k,\infty} = A_{k}+G_{k}K_{k,\infty}\\
&E_{A_{k}} + E_{B_{k}}K_{k,\infty} = 0,
\end{aligned}
\end{equation}
%--
and a sufficient condition that satisfies (\ref{feedback}) is
%--
\begin{equation}
rank\,\big(\begin{bmatrix}
E_{A_{k}} & E_{B_{k}}
\end{bmatrix}\big) = rank\,\big(E_{B_{k}}\big).
\end{equation}
%--
See the paper of \cite{Terra2014} for more details on convergence and stability analysis. Lastly, the \autoref{rlqr_algor_ap} shows the RLQR obtained with \autoref{theorem_rlqr}.

%---------------------------------------------------------------------------------------------------------------------------------------
\begin{algorithm*}[h]
	\textbf{Uncertain model:} Consider the model (\ref{eq:statespaceDT})-(\ref{eq:RLQRuncertainties}) and criterion (\ref{eq:rlqrminmax})-(\ref{eq:rqlqrcostfun}) with known\\ $A_{k}$, $B_{k}$, $E_{A_{k}}$, $E_{B_{k}}$, $Q_{k} \succ 0$, and $R_{k} \succ 0$ for all $k$.\\
	\textbf{Initial conditions:} Define $x_{0}$ and $P_{k,N}\succ{0}$.\\
	\textbf{Step 1:} \textit{(Backward)} For all $k=N-1,\ldots,0$, compute\\
	$\hfill\begin{bmatrix}
	L_{k}\\
	K_{k}\\
	P_{k}
	\end{bmatrix}
	=
	\begin{bmatrix}
	0 & 0 & 0 \\
	0 & 0 & 0 \\
	0 & 0 & -I \\
	0 & 0 & A_{k} \\
	0 & 0 & E_{A_{k}} \\
	I & 0 & 0 \\
	0 & I & 0
	\end{bmatrix}^{T}
	\begin{bmatrix}
	P_{k+1}^{-1}& 0 & 0 & 0 & 0 & I & 0\\
	0 & R_{k}^{-1} & 0 & 0 & 0 & 0 & I\\
	0 & 0 & Q_{k}^{-1} & 0 & 0 & 0 & 0\\
	0 & 0 & 0 & 0 & 0 & I & -B_{k}\\
	0 & 0 & 0 & 0 & 0 & 0 & -E_{B_{k}}\\
	I & 0 & 0 & I & 0 & 0 & 0\\
	0 & I & 0 & -B_{k}^{T} & -E_{B_{k}}^{T} & 0 & 0
	\end{bmatrix}^{-1}
	\begin{bmatrix}
	0\\
	0\\
	-I\\
	A_{k}\\
	E_{A_{k}}\\
	0\\
	0
	\end{bmatrix}.\hfill$\\ \vspace{0.1cm}
	\textbf{Step 2:} \textit{(Forward)} For each $k=0,...,N-1$, obtain
	\\
	$\hspace*{\fill} \begin{bmatrix}
	x^{*}_{k+1}\\
	u^{*}_{k}
	\end{bmatrix}
	=
	\begin{bmatrix}
	L_{k}\\
	K_{k}
	\end{bmatrix}
	x^{*}_{k},\hspace*{\fill}$ \\
	\\
	with the total cost given by $J_{r}^{*}=x_{0}^{T}P_{0}x_{0}$.
% 	\\
	\caption{The Robust Linear Quadratic Regulator}
	\label{rlqr_algor_ap}
\end{algorithm*}

\section{Uncertainties matrices}
\label{sec:uncertainties}

The uncertainties in (\ref{eq:uncertainties}) can be estimated based on the inertial variation related to minimum and maximum payload values, given by $m_{p_{max}}$ and $m_{p_{min}}$. Thus, the parameter values of the matrices $F$ and $G$ can vary according to 
\begin{equation}
\Gamma_{F} = F_{m_{p_{min}}}-F_{m_{p_{max}}},
\end{equation}
\begin{equation}
\Gamma_{G} = G_{m_{p_{min}}}-G_{m_{p_{max}}},
\end{equation}
where $F_{m_{p_{min}}}$, $G_{m_{p_{min}}}$, $F_{m_{p_{max}}}$ and $G_{m_{p_{max}}}$ are the discretized state-space matrices in (\ref{eq:estate-space_rewritten}) corresponding to $m_{p_{max}}$ and $m_{p_{min}}$. 
Next, the column that is most affected by the mass fluctuations is chosen, in this case, the second one.

The considered uncertainty range corresponds to the fully loaded and unloaded vehicle operation, $m_{p_{max}}$ and $m_{p_{min}}$, respectively. Thus, the matrices $E_{F}$ and $E_{G}$ are then calculated as:
%--
\begin{equation}
E_{F} = 
\begin{bmatrix}
1\\
1\\
1\\
0.1\\
\end{bmatrix}^T
\begin{bmatrix}
\underset{\Gamma_{F_{1,2}}}{arg}(\abs{\Gamma_{F_{1,2}}}) & 0 & 0 & 0 \\ 0 & \underset{\Gamma_{F_{2,2}}}{arg}(\abs{\Gamma_{F_{2,2}}}) & 0 & 0\\ 0 & 0 & \underset{\Gamma_{F_{3,2}}}{arg}(\abs{\Gamma_{F_{3,2}}}) & 0 \\ 0 & 0 & 0 & \underset{\Gamma_{F_{4,2}}}{arg}(\abs{\Gamma_{F_{4,2}}})
\end{bmatrix}
\end{equation}
%--
Following this, select the largest value in $\Gamma_G$ and obtain $E_G$ as
%--
\begin{equation}
\label{Eg_calc}
E_{G} = 
\begin{bmatrix}
0.1 \\ 
0.1
\end{bmatrix}^T
\begin{bmatrix}
\underset{\Gamma_{G_{j,1}}}{arg} \{\max(\abs{\Gamma_{G_{j,1}}})\} & 0\\
0 & \underset{\Gamma_{G_{j,1}}}{arg} \{\max(\abs{\Gamma_{G_{j,1}}})\}
\end{bmatrix}
\end{equation}
%--
and $H = [1,1,1,1]^T$. Finally, the uncertainties matrices are obtained through (\ref{eq:RLQRuncertainties}) as
%--
\begin{equation*}
\begin{bmatrix}
\delta {F} & \delta G
\end{bmatrix} = H \Delta \begin{bmatrix}
E_{F} & E_{G}
\end{bmatrix},
\end{equation*}
%--
where $\Delta$ is a scalar represented by the mass variation.

% An appendix contains supplementary information that is not an essential part of the text itself but which may be helpful in providing a more comprehensive understanding of the research problem or it is information that is too cumbersome to be included in the body of the paper.




\end{appendices}


\bibliography{ref}% common bib file


\end{document}